\documentclass[american,aps,pra, reprint]{revtex4-1}
\usepackage[T1]{fontenc}
\usepackage[latin9]{inputenc}
\usepackage{babel}
\usepackage{amsthm}
\usepackage{amsmath}
\usepackage{bm}
\usepackage{amssymb}
\usepackage{graphics,graphicx}
\usepackage{tikz}
\usepackage{float}
\usepackage[unicode=true,pdfusetitle, bookmarks=true,bookmarksnumbered=false,bookmarksopen=false,
 breaklinks=false,pdfborder={0 0 0},backref=false,colorlinks=false]
 {hyperref}
\hypersetup{colorlinks,linkcolor=myurlcolor,citecolor=myurlcolor,urlcolor=myurlcolor}

\makeatletter
\@ifundefined{textcolor}{}
{%
 \definecolor{BLACK}{gray}{0}
 \definecolor{WHITE}{gray}{1}
 \definecolor{RED}{rgb}{1,0,0}
 \definecolor{GREEN}{rgb}{0,1,0}
 \definecolor{BLUE}{rgb}{0,0,1}
 \definecolor{CYAN}{cmyk}{1,0,0,0}
 \definecolor{MAGENTA}{cmyk}{0,1,0,0}
 \definecolor{YELLOW}{cmyk}{0,0,1,0}
 }
\theoremstyle{plain}
\newtheorem{thm}{\protect\theoremname}
\ifx\proof\undefined
\newenvironment{proof}[1][\protect\proofname]{\par
\normalfont\topsep6\p@\@plus6\p@\relax
\trivlist
\itemindent\parindent
\item[\hskip\labelsep
\scshape
#1]\ignorespaces
}{%
\endtrivlist\@endpefalse
}
\providecommand{\proofname}{Proof}
\fi

\usepackage{times}
\usepackage{txfonts}
\usepackage{braket}
\usepackage{colortbl}
\definecolor{myurlcolor}{rgb}{0,0,0.7}

\makeatother

\providecommand{\theoremname}{Theorem}
\newcommand{\ketbra}[2]{|{#1} \rangle \langle {#2} |}
\usepackage{times}

\usepackage[up]{subfigure}

\begin{document}
\title{Maximally coherent mixed states: Complementarity between maximal coherence and mixedness}

\begin{abstract}
Quantum coherence is a key element in topical research on quantum resource theories and a primary facilitator for design and implementation of quantum technologies. However, the resourcefulness of quantum coherence is severely restricted by environmental noise, which is indicated by the loss of information in a quantum system, measured in terms of its purity. In this work, we derive the limits imposed by the mixedness of a quantum system on the amount of quantum coherence that it can possess. We obtain an analytical trade-off between the two quantities that upperbound the maximum quantum coherence for fixed mixedness in a system. This gives rise to a class of quantum states, ``maximally coherent mixed states,'' whose coherence cannot be increased further under any purity-preserving operation. For the above class of states, quantum coherence and mixedness satisfy a complementarity relation, which is crucial to understand the interplay between a resource and noise in open quantum systems. 

\end{abstract}
 
 \author{Uttam Singh}
 
 \author{Manabendra Nath Bera}

 \author{Himadri Shekhar Dhar}
 
  \author{Arun Kumar Pati}
 
 \affiliation{
 Harish-Chandra Research Institute, Allahabad-211019, India}


\maketitle

\section{Introduction}
Recent developments in modern science have shown that quantum coherence plays an important role in low-temperature physics starting from the formulation of the basic laws of thermodynamics to work extraction
\cite{Horodecki2013, Aspuru13, Skrzypczyk2014, Rudolph114, Rudolph214, Varun14, Aberg14, Oppenheim14, Correa2014, Abah2014, Brandao2015}.
Furthermore, it is a useful figure of merit in investigating nanoscale systems \cite{VazquezH2012, Karlstrom11} and understanding efficient energy transfer in complex biological systems 
\cite{Aspuru2009, Lloyd2011, Li2012, Huelga13, Levi14}. 
In recent years, researchers have attempted to develop a framework to formalize the theory of quantum coherence within the realms of quantum information and quantum resource theories \cite{Gour2008, Marvian2013, Fernando2013, Marvian14, Baumgratz2014, Adesso2015, Girolami14, Marvian2014, Fan2014, FBrandao15, Fan15, Pinto2015}. 
Within this context, there are two pertinent theoretical frameworks that attempt to characterize coherence as a resource. The first is based on the resource theory of asymmetry relative to phase shifts, where operations are restricted to phase insensitive operations and symmetric states are free resources \cite{Gour2008, Marvian2013, Marvian2014, Marvian14}. The above theoretical structure has been used in several resource based formulations of quantum thermodynamics \cite{Rudolph114, Rudolph214}. 
The second formalism 
is based on a well-defined set of allowed incoherent operations and a set of freely available incoherent states \cite{Baumgratz2014}. In this framework, quantum coherence is a well-defined resource, which can be quantified in terms of functions or coherence monotones that satisfy certain characteristic conditions. 
Some of the better known measures of quantum coherence are those based on $l_1$ norm and relative entropy \cite{Baumgratz2014}, and skew information \cite{Girolami14}. Incidentally, a recent work proves that all measures of entanglement can be artfully used to define a family of valid measures of quantum coherence \cite{Alex15}. Moreover, the latter formalism has been recently used to address a fundamental issue of wave-particle duality \cite{Manab2015}, thus, enabling coherence to be a valid indicator of the wave 
nature of quantum systems.

Another significant aspect in the dynamics of quantum systems is the role of environmental noise and the unavoidable phenomenon of decoherence. It is known that decoherence is detrimental to the amount of information contained in a quantum state, as measured by its purity. 
To effectively characterize the role of decoherence in erasing information \cite{Landauer1961} one needs to quantify the purity or its complementary property, the mixedness of the state. A faithful measure of mixedness is the normalized linear entropy \cite{Peter04}. From the perspective of resource theory of purity \cite{Oppenheim03, Oppenheim13}, mixedness can be obtained as a complementary quantity to global information. Since, noise tends to increase the mixedness of a quantum system, it emerges as an intuitive parameter to understand decoherence. A natural question that arises is how does important physical quantities in quantum information theory, such as entanglement \cite{HorodeckiRMP09}, fare against mixedness of quantum systems? An interesting direction is to obtain the maximum amount of entanglement for a given mixedness, which leads to the notion of maximally entangled mixed states  \cite{Ishizaka00, Munro01, Verstraete01, Verstraete03, Kwiat04}. The amount of entanglement in such states cannot be increased further under any global unitary operation. Also,  the form of the maximally entangled mixed states depends on the measures employed to quantify entanglement and mixedness in the system \cite{Verstraete03}. Such states have also been investigated in Gaussian quantum systems \cite{Adesso2004, Adesso04, Adesso2007}.

In our work, we investigate the limits imposed by mixedness of a quantum system on the amount of quantum coherence present in the system. 
Since we consider quantum systems where the missing phase-reference frame is apparently lacking, the formalism based on the resource theory of asymmetry \cite{Marvian14} becomes over-restrictive 
\footnote{The discussion on coherence, in Ref.~\cite{Marvian14}, states that the set of phase insensitive operations restricts transformations allowed by incoherent operations. However, restrictive operations are not necessarily physically more relevant or significant under all contexts. For example, in a recent article on the resource theories of thermodynamics based on Gibbs preserving operations and thermal operations \cite{Faist2015}, the authors have argued how Gibbs preserving operations outperform thermal operations even though the former is restricted by the latter set of operations.}.
%
Hence, in the present work, we use the theoretical approach based on the set of incoherent operations and states \cite{Baumgratz2014}, to characterize and quantify coherence.
We derive an analytical trade-off between the two quantities that allows us to upperbound the maximum coherence in a given mixed quantum state and vice versa. Using the $l_1$ norm of coherence \cite{Baumgratz2014} as a measure of quantum coherence and normalized linear entropy \cite{Peter04} as a measure of mixedness, we prove that for a general $d$-dimensional quantum system the sum of the (scaled) squared coherence and the mixedness is always less than or equal to unity. This allows us to derive a class of quantum states, viz. ``maximally coherent mixed states'' (MCMSs), that have maximal coherence, up to incoherent unitaries, for a fixed mixedness. These states are parametrized
mixtures of a $d$-dimensional pure maximally coherent state and maximally mixed state. Interestingly, for different values of mixedness the analytical form of MCMS remains unchanged and, unlike maximally entangled mixed states, is not dependent on the choice of the measure of coherence and mixedness, as observed for $l_1$ norm, relative entropy, and geometric measures of coherence.
The obtained analytical results, show an important trade-off between a relevant quantum resource and noise in open quantum systems and a complementary behavior between coherence and mixedness in the class of MCMSs, which may be crucial from the perspective of quantum resource theories and thermodynamics. 
Significantly, since the mixedness of a quantum system can be experimentally measured using quantum interferometric setups \cite{Erik2000, Ekert2002}, without resorting to complicated state tomography, our results provide a mathematical framework to experimentally determine the maximal coherence in a quantum state.
%

The paper is organized as follows. In Sec. \ref{meas}, we briefly discuss the quantification of coherence and mixedness. In Sec. \ref{trade}, we theorize the trade-off between coherence and mixedness in $d$-dimensional systems. In Sec. \ref{comp}, we define a class of maximally coherent mixed states that satisfy a complementarity relation between coherence and mixedness. In Sec. \ref{trans}, we investigate the allowed set of transformations within classes of fixed coherence or mixedness. We conclude with a discussion of the main results in Sec. \ref{dis}.

\section{Quantifying coherence and mixedness}
\label{meas}
In this section we present a brief overview of the concepts of quantum coherence and mixedness of quantum systems. To characterize the coherence in a quantum system, we follow the theoretical approach developed in Ref.~\cite{Baumgratz2014}. 
All mathematical formulations and results that are subsequently presented and discussed are valid within the framework of the above theory of quantum coherence.


\subsection{Quantum coherence}
Quantum coherence, an essential feature of quantum mechanics arising from the superposition principle, is inherently a basis dependent quantity. Therefore, any quantitative measure of it must depend on a reference basis. 
The framework, to quantify coherence in the context of quantum information theory, is based on the characterization of a set of incoherent states, denoted by $\mathcal{I}$ and incoherent operations $\Lambda^\mathcal{I}$ \cite{Baumgratz2014}. For a given reference basis $\{\ket{i}\}$, all the states of the form $\rho_I = \sum_{i}d_i \ketbra{i}{i}$, where $d_i\geq 0$ and $\sum_{i}d_i=1$,  form a set, $\mathcal{I}$, of incoherent states. Incoherent operations $\Lambda^\mathcal{I}$ are defined as completely positive trace preserving (CPTP) maps, which map the set of incoherent states onto itself, i.e., $\Lambda^\mathcal{I}(\mathcal{I}) \in \mathcal{I}$. Under the set of operations $\Lambda^\mathcal{I}$ and the free incoherent states $\mathcal{I}$, quantum coherence is a valid resource that can be quantified. 
A function, $C(\rho)$, is a bona fide measure of quantum coherence of the state $\rho$ if it satisfies the following conditions \cite{Baumgratz2014} :
(1) $C(\rho) = 0$ iff $\rho \in \mathcal{I}$.
(2) $C(\rho)$ is nonincreasing under the incoherent operations, i.e., $C(\Lambda_I[\rho]) \leq C(\rho)$.
(3) $C(\rho)$ decreases on an average under the selective incoherent operations, i.e., $\sum _k p_k  C(\rho_k) \leq C(\rho)$, where $\rho_k = M_k \rho M^\dagger_k/p_k$, $p_k = \mathrm{Tr}M_k \rho M^\dagger_k$, and $M_k$ are the Kraus elements of an incoherent channel.
(4) $C(\rho)$ is convex in its arguments, i.e., $C(\sum_k p_k \rho_k) \leq \sum_k p_k C(\rho_k)$.
One may note that conditions (3) and (4) together imply condition (2). 

Measures that satisfy the above conditions, include 
$l_1$ norm and relative entropy of coherence \cite{Baumgratz2014} and the skew information \cite{Girolami14}. Generic monotones of quantum coherence can also be derived using entanglement monotones that satisfy the above conditions \cite{Alex15}.
In this work, we shall mainly be focused on the $l_1$ norm of coherence.
For a quantum state $\rho$ and the reference basis $\{\ket{i}\}$, the $l_1$ norm of coherence is given by
\begin{align}
 C_{l_1}(\rho) = \sum_{i\neq j} |\rho_{ij}|,
\end{align}
where $\rho_{ij}=\langle i | \rho | j\rangle$. Another measure of coherence is the relative entropy of coherence, which is given by 
$
C_{r}(\rho) = S(\rho_d) - S(\rho),
$
where $S(\rho)=-\mathrm{Tr} (\rho~ \mbox{ln}~ \rho)$, is the von Neumann entropy and $\rho_d=\sum_i \langle i | \rho | i\rangle | i\rangle\langle i |$. Moreover, a geometric measure of coherence had also been speculated \cite{Baumgratz2014, Fan15, Alex15} and was, only recently, shown to be a full coherence monotone \cite{Alex15}. The geometric measure is given by $C_g(\rho) = 1-\max_{\sigma\in\mathcal{I}} F(\rho,\sigma)$, where $\mathcal{I}$ is the set of all incoherent states and $F(\rho,\sigma) = \left(\mathrm{Tr}[\sqrt{\sqrt{\sigma}\rho\sqrt{\sigma}}]\right)^2$ is the fidelity of the states $\rho$ and $\sigma$.
It is important to note that quantum coherence, by definition, is not invariant under general unitary operation but does remain unchanged under incoherent unitaries. Furthermore, the maximally coherent pure state is defined by $\ket{\psi_d} = \frac{1}{\sqrt{d}}\sum_{i=0}^{d-1} \ket{i}$, for which $C_{l_1}(\ket{\psi_d}\bra{\psi_d}) = d-1$ and $C_r(\ket{\psi_d}\bra{\psi_d}) = \mbox{ln}~ d$. 

\subsection{Mixedness}
For every quantum state, the ubiquitous interaction with environment or decoherence affects its purity. Noise introduces mixedness in the quantum system leading to loss of information, and 
hence, its characterization is an important task in quantum information protocols. 
The mixedness, which represents nothing but the disorder in the system, can be quantified in terms of entropic functionals, such as linear and von Neumann entropy of the quantum state. For an arbitrary $d$-dimensional state, the mixedness, based on normalized linear entropy \cite{Peter04}, is given as  
\begin{align}
 M_l(\rho) = \frac{d}{d-1}\left( 1 - \mathrm{Tr}\rho^2 \right).
\end{align}
Therefore, for each quantum system, mixedness varies between $0$ and $1$, i.e., $0\leq M_l(\rho)\leq 1$. Furthermore, since $\mathrm{Tr}\rho^2$ describes the purity of quantum system, mixedness expectedly emerges as a complementary quantity to the purity of the given quantum state.
%
%
%
The other operational measure of mixedness of a quantum state $\rho$ is the von Neumann entropy, $S(\rho)=-\mathrm{Tr} (\rho~ \mbox{ln}~ \rho)$. Moreover, in a manner similar to quantum coherence, a geometric measure of mixedness can also be defined, which is given by $M_g(\rho) := F(\rho,\mathbb{I}/d) = \frac{1}{d}\left( \mathrm{Tr}\sqrt{\rho} \right)^2$ and lies between $0$ and $1$.

\section{Trade-off between quantum coherence and mixedness}
\label{trade}

In this section, we investigate the restrictions imposed by the mixedness of a system on the maximal amount of quantum coherence. We prove analytically, that there exists a trade-off between the two quantities and for a fixed amount of mixedness the maximal amount of coherence is limited. The results allow us to derive a class of states that are the most resourceful, in terms of quantum coherence, under a fixed amount of noise, characterized by its mixedness.

The important trade-off between quantum coherence, as quantified by the $l_1$ norm, and mixedness, in terms of the normalized linear entropy, is captured by the following theorem.
\begin{thm}
\label{t1}
For any arbitrary quantum system, $\rho$, in $d$ dimensions, the amount of quantum coherence $C_{l_1}(\rho)$ in the state is restricted by the amount of mixedness $M_l(\rho)$ through the 
the inequality
\begin{align}
\label{Gconj}
 \frac{C_{l_1}^2(\rho)}{(d-1)^2} + M_l(\rho) \leq 1.
\end{align}
\end{thm}
\begin{proof}
Using the parametric form of an arbitrary density matrix, 
the state of a $d$-dimensional quantum system can be written in terms
of the generators, $\hat{\Lambda}_i$, of $SU(d)$ \cite{Nielsen10, Eberly81, Mahler98, Kimura03, Khaneja03}, as
\begin{align}
\label{qd}
 \rho = \frac{\mathbb{I}}{d} + \frac{1}{2}\sum_{i=1}^{d^2-1}x_i \hat{\Lambda}_i,
\end{align}
where $x_i = \mathrm{Tr}[\rho \hat{\Lambda}_i]$. 
The condition of positivity can be stated in terms of the coefficients of the characteristic equation for the density matrix $\rho$. Specifically, the Eq. (\ref{qd}) is positive iff all the coefficients of the polynomial det$(\lambda \mathbb{I}-\rho) = \sum_{i=0}^{d}(-1)^iA_i\lambda^{d-i}$ = 0, $A_i\geq 0$ for $1\leq i\leq d$ ($A_0=1$). This criterion can be verified simply by calculating traces of various powers of $\rho$ \cite{Kimura03,Khaneja03}.
The generators $\hat{\Lambda}_i$ $(i=1,2,...,d^2-1)$ satisfy $(1)$  $\hat{\Lambda}_{i} = \hat{\Lambda}_{i}^\dag$,
$(2)$ Tr$(\hat{\Lambda}_{i})=0$, and $(3)$ Tr$(\hat{\Lambda}_{i}\hat{\Lambda}_{j})=2\delta{ij}$.
These generators are defined by the structure constants $f_{ijk}$ (a completely antisymmetric tensor) and
$g_{ijk}$ (a completely symmetric tensor), of Lie algebra $su(d)$ \cite{Mahler98, Kimura03}.
The generators can be conveniently written as
$\{\hat{\Lambda}_i\}_{i=1}^{d^2-1} = \{\hat{u}_{jk}, \hat{v}_{jk}, \hat{w}_{l}\}$. Here
 $\hat{u}_{jk} = (\ketbra{j}{k}+\ketbra{k}{j})$,
 $\hat{v}_{jk} = -i(\ketbra{j}{k}-\ketbra{k}{j})$,
and 
 $\hat{w}_{l} = \sqrt{\frac{2}{l(l+1)}} \sum_{j=1}^{l}\left(\ketbra{j}{j}-l\ketbra{l+1}{l+1}\right)$,
where $j<k$ with $j,k=1,2,...,d$ and $l=1,2,...,(d-1)$ \cite{Mahler98, Kimura03}.
The generators can be labeled as $\{\hat{\Lambda_1},..,\hat{\Lambda}_{\frac{(d^2-d)}{2}}, \hat{\Lambda}_{\frac{(d^2-d)}{2}+1},..,
\hat{\Lambda}_{(d^2-d)}, \hat{\Lambda}_{(d^2-d)+1},..,\hat{\Lambda}_{(d^2-1)} \} 
= \{ \hat{u}_{12},..,\hat{u}_{(d-1)d}, \hat{v}_{12},..,\hat{v}_{(d-1)d}, \hat{w}_{1},..,\hat{w}_{(d-1)} \}$.

The $l_1$ norm of coherence of a $d$-dimensional system, given by  Eq.~(\ref{qd}), can be written as
\begin{align}
 C_{l_1}(\rho) 
 & = \sum_{i=1}^{(d^2-d)/2}\sqrt{x_i^2 + x_{i+(d^2-d)/2}^2 }.
\label{l1eq}
\end{align}
Furthermore, the mixedness is given by
\begin{align}
 M_l(\rho) = \frac{d}{d-1}(1 - \mathrm{Tr}\rho^2) = 1 - \frac{d}{2(d-1)}\sum_{i=1}^{d^2-1}x_i^2.
\label{mixeq}
\end{align}
Using the expressions for $C_{l_1}(\rho)$ and $M_l(\rho)$, we obtain
\begin{align}
 &\frac{C_{l_1}^2(\rho)}{(d-1)^2} + M_l(\rho)\nonumber\\
 &= \frac{1}{(d-1)^2} \left(\sum_{i=1}^{(d^2-d)/2}\sqrt{x_i^2 + x_{i+(d^2-d)/2}^2 }\right)^2 
 + 1 - \frac{d}{2(d-1)}\sum_{i=1}^{d^2-1}x_i^2\nonumber\\
  &= 1 - \frac{1}{(d-1)^2} \sum_{i=1}^{d^2-1}x_i^2\nonumber\\
  &+ \frac{1}{(d-1)^2} \left(\left(\sum_{i=1}^{(d^2-d)/2}\sqrt{x_i^2 
  + x_{i+(d^2-d)/2}^2 }\right)^2 - (\frac{d^2-d}{2}-1) \sum_{i=1}^{d^2-1}x_i^2\right)\nonumber\\
 %
 &= 1 - \frac{1}{(d-1)^2} \sum_{i=1}^{d^2-1}x_i^2 - \frac{((d^2-d)/2-1)}{(d-1)^2} \sum_{i=d^2-d}^{d^2-1}x_i^2\nonumber\\ 
 &~+ \frac{1}{(d-1)^2} \left(\left(\sum_{i=1}^{(d^2-d)/2}\sqrt{x_i^2 
 + x_{i+(d^2-d)/2}^2 }\right)^2 - (\frac{d^2-d}{2}-1) \sum_{i=1}^{d^2-d}x_i^2 \right)\nonumber\\
 %
 %
 %
 &\leq 1 - \frac{d}{2(d-1)} \sum_{i=d^2-d}^{d^2-1}x_i^2,
\end{align}
where, in the last step, we have used the inequality $2\sqrt{xy}\leq (x+y)$. Since the $\frac{d}{2(d-1)} \sum_{i=d^2-d}^{d^2-1}x_i^2 \geq 0$,  we have $\frac{C_{l_1}^2(\rho)}{(d-1)^2} + M_l(\rho) \leq 1$, which concludes our proof.
\end{proof}

Theorem \ref{t1} proves that the scaled coherence, $\frac{C_{l_1}(\rho)}{(d-1)}$, of a quantum system with mixedness $M_l(\rho)$, is bounded to a region below the parabola $\frac{C_{l_1}^2(\rho)}{(d-1)^2} + M_l(\rho) = 1$ (see Fig. \ref{fig}).
The quantum states with (scaled) quantum coherence that lie on the parabola are the maximally coherent states corresponding to a fixed mixedness and vice versa. The trade-off obtained between coherence and mixedness can be neatly presented for a qubit system.
Let us consider an arbitrary single-qubit density matrix of the form
\begin{equation}
\label{state}
\rho=\left(\begin{array}{cc}
 a & c\\
 c^* & 1-a
\end{array} \right).
\end{equation}
The eigenvalues of the above density matrix are given by $\lambda_{\pm} = \left(1 \pm \sqrt{1-4[a(1-a)-4|c|^2]}\right)/2$.
The positivity and Hermiticity of the density matrix implies that $0\leq a(1-a)-4|c|^2\leq 1/4$. Now, the mixedness of the 
state $\rho$ is given by
$
M_l(\rho) = 4a(1-a)- 4|c|^2.
$
The $l_1$ norm of coherence is 
$C_{l_1}(\rho) = 2|c|$.
Using the expressions of coherence and mixedness, we obtain
$
C_{l_1}^2(\rho) +  M_l(\rho) = 4a(1-a).
$
Since $4a(1-a)\leq 1 $, we have 
$
C_{l_1}^2(\rho) + M_l(\rho) \leq 1,
$
with the equality holding if and only if $a=1/2$. 

From Theorem \ref{t1}, we know that the maximum coherence permissible in an arbitrary quantum state with a fixed mixedness, are the values that lie on the parabola $\frac{C_{l_1}^2(\rho)}{(d-1)^2} + M_l(\rho) = 1$. The same holds for the maximum mixedness allowed in a quantum state with fixed coherence (see Fig. \ref{fig}).
A natural question arises: What are the quantum states that correspond to the maximal coherence and satisfy the equality in Eq. (\ref{Gconj})? The above question is addressed in the following section.

\section{Maximally coherent mixed states and complementarity}
\label{comp}
Let us find the quantum states with maximal $l_1$ norm of coherence for a fixed amount of mixedness, say $M_f$. For this, we need to maximize the coherence under the constraint that the mixedness $M_f$ as quantified by normalized linear entropy is invariant.
Here we provide the form of maximally coherent mixed state for a general $d$-dimensional system.
\begin{thm}
\label{t2}
An arbitrary $d$-dimensional quantum system with maximal coherence for a fixed mixedness $M_f$, up to incoherent unitaries, is of the following form
\begin{align}
\label{mxc}
 \rho_m = \frac{1-p}{d}\mathbb{I}_{d\times d} + p~ \ketbra{\psi_d}{\psi_d},
\end{align}
where $\ket{\psi_d} = \frac{1}{\sqrt{d}}\sum_{i=1}^{d}\ket{i}$, is the maximally coherent state in the computational basis, $\mathbb{I}_{d\times d}$ is the $d$-dimensional identity operator and the mixedness, in terms of normalized linear entropy, is equal to $M_f$ = $1-p^2$.
\end{thm}

\begin{proof}
Using the parametric form of the density matrix given in Eq. (\ref{qd}), the expressions for coherence and mixedness of any $d$-dimensional system was obtained in Eqs. (\ref{l1eq}) and (\ref{mixeq}). 
To prove the above theorem, we seek the maximal coherence for a fixed mixedness, say $M_f$, i.e., we maximize the function $ C_{l_1}$, under the constraint
\begin{align}
\label{const}
 M_f = 1 - \frac{d}{2(d-1)}\sum_{i=1}^{d^2-1}x_i^2.
\end{align}
Hence, we need to maximize the Lagrange function
\begin{align}
\mathcal{L} = \sum_{i=1}^{D/2}\sqrt{x_i^2 + x_{i+D/2}^2 } + \lambda \left(1 - \frac{d}{2(d-1)}\sum_{i=1}^{D+d-1}x_i^2 - M_f\right),
\end{align}
where $D=d^2-d$ and $\lambda$ is the Lagrange multiplier.
The stationary points, $\left\{x'_j\right\}$, of $C_{l_1}(\rho)$ imply the vanishing of 
\begin{align}
  \frac{\partial \mathcal{L}}{\partial x_j}\left|_{\left\{x'_j\right\}} \right.= \left\{
  \begin{array}{l l}
    \frac{x'_j}{\sqrt{x'^2_j+ x'^2_{j+D/2}}} -\frac{\lambda d}{d-1}{x'_j}~, & \quad \text{for $j\leq D/2$}\\
    -\frac{\lambda d}{d-1}{x'_j}~, & \quad \text{for $j>D$}
  \end{array} \right..
\end{align}
Therefore, we have ${x'_j} = 0$ for all $j>D$ and $\sqrt{x'^2_j+ x'^2_{j+D/2}} = \frac{d-1}{\lambda d}$ for $j\leq D/2$.
This implies that
\begin{align}
  x'^2_1+x'^2_{1+D/2}  = x'^2_2+x'^2_{2+D/2}  = \cdots  = x'^2_{D/2}+x'^2_{D}  = \left(\frac{d-1}{\lambda d}\right)^2.
\end{align}
Putting these values of ${x'_j}$'s in the constraint equation (\ref{const}) we get
$ \lambda =(d-1)/[2\sqrt{(1-M_f)}]$. The positive value of $\lambda$ is chosen because negative value
leads to negative coherence, which is not desired. The value of coherence for the stationary states is given by 
\begin{align}
 C_{l_1}(\rho) = \sum_{j=1}^{D/2}\sqrt{x'^2_j + x'^2_{j+D/2} } = (d-1) \sqrt{(1-M_f)}.
\end{align}
This is the maximal value of coherence that a state can have for a fixed value of mixedness $M_f$. Therefore,
the states with $x_j^2+x_{j+D/2}^2 = 4(1-M_f)/ d^2$ for $j\leq D/2$ and $x_j = 0$ for $j>D$
are the states that have maximum coherence for a given mixedness $M_f$. These states can
be written as
\begin{align}
 \rho_m &= \frac{\mathbb{I}}{d} + \frac{R}{2} \sum_{i=1}^{D/2} (\cos\theta_i \hat{\Lambda}_i + \sin\theta_{i} \hat{\Lambda}_{i+D/2}),
\end{align}
where $R = \frac{2\sqrt{(1-M_f)}}{d}$ and $\theta_i=\tan^{-1}(x_{i+D/2}/x_i)$. We observe that the diagonal part of these states is maximally mixed and the points, $\left\{x_i,x_{i+D/2}\right\}_{i = 1}^{D/2}$, that define the off-diagonal elements, lie on the circle of radius $R$ in the real $(x_i,x_{i+D/2})$--plane. 
An equivalent form of the above states can be written, by identifying $\{\theta_1,...,\theta_{d-1},\theta_{d},...,\theta_{(d^2-d)/2}\} 
= \{ \phi_{12},...,\phi_{1d}, \phi_{23},...,\phi_{(d-1)d} \}$, as
\begin{align}
 \rho_m = \frac{\mathbb{I}}{d} + \frac{R}{2} \sum_{\substack{   {i,j=1} \\   i< j   }}^d (e^{i\phi_{ij}} \ketbra{i}{j} + e^{-i\phi_{ij}} \ketbra{j}{i}).
\end{align}
Now, the phases appearing in the off diagonal components can be removed by applying an incoherent unitary of the form $U = \sum_{n=1}^{d} e^{-i\gamma_{n}}\ketbra{n}{n}$, 
which keeps the coherence invariant. To this end by choosing $\phi_{ij} = \gamma_i - \gamma_j$ we get
\begin{align}
 \rho_m = \frac{\mathbb{I}}{d} + \frac{R}{2} \sum_{\substack{   {i,j=1} \\   i< j   }}^d ( \ketbra{i}{j} + \ketbra{j}{i}).
\end{align}
Now, setting $R = 2p/d$, we obtain the state given in Eq. (\ref{mxc}). Therefore, up to incoherent unitary transformations, 
the states with maximal coherence for a fixed mixedness are those that take the form given by Eq. (\ref{mxc}). This completes the proof.
\end{proof}

For a single-qubit quantum system, the proof can be mathematically elaborated. For the density matrix, given in Eq. (\ref{state}), we need to maximize the coherence under the constraint that $M_f = 4a(1-a) - 4|c|^2$, is invariant. Hence, we need to maximize,
$
 C_{l_1}(\rho) = 2 |c| + \lambda [4a(1-a) - 4|c|^2 -M_f],
$
where $\lambda$ is the Lagrange multiplier. Upon optimization, the stationary points are given by $a=1/2$ and $|c| = 1/(4\lambda)$. Using constraint equation, we get $\lambda = \pm 1/(2\sqrt{1-M_f})$. Choosing the positive value of $\lambda$, we obtain $|c| =\sqrt{1-M_f}/2 $. Thus, the
maximum value of coherence is equal to $C_{l_1}(\rho) = \sqrt{1-M_f}$ and the corresponding states, are given by
\begin{equation}
\label{cmax}
\rho_{m}(\phi)=\frac{1}{2}\left(\begin{array}{cc}
 1 & \sqrt{1-M_f}\exp[i \phi]\\
 \sqrt{1-M_f}\exp[-i \phi] & 1
\end{array} \right),
\end{equation}
where $\phi$ is an arbitrary phase. The phase can be removed through incoherent unitaries which keeps the coherence invariant. The density matrix in Eq. (\ref{cmax}), up to incoherent unitaries, has the form 
$
 \rho_{m}=\frac{1-p}{2}{\mathbb{I}}_{2\times2} + p \ketbra{\psi_2}{\psi_2},
$
where $\ket{\psi_2} = (\ket{0} + \ket{1})/\sqrt{2}$ is the maximally coherent state,  $\mathbb{I}_{2\times2}$ is the identity operator in two dimensions, and $p = \sqrt{1-M_f}$. 

From Theorem \ref{t2}, the $l_1$ norm of coherence of the maximally coherent mixed state, given in Eq. (\ref{mxc}), is $C_{l_1}(\rho_m) = (d-1)p$, and
the mixedness is equal to $M_l(\rho_m) = \frac{d}{d-1}\left(1-\mathrm{Tr}[\rho_m^2]\right) = 1- p^2$.
Therefore, we obtain a \emph{complementarity} relation between coherence and mixedness, 
\begin{align}
\label{conj}
 \frac{C_{l_1}^2(\rho_m)}{(d-1)^2} + M_l(\rho_m) =1,
\end{align}
which satisfies the equality in Eq.~(\ref{Gconj}), and thus lie on the parabola, $\frac{C_{l_1}^2(\rho_m)}{(d-1)^2} + M_l(\rho_m) =1$, in the coherence-mixedness plane (see Fig.~\ref{fig}). We call the parametrized class of states, defined by Eq. (\ref{mxc}), that satisfy the complementarity between coherence and mixedness, i.e., any change in coherence leads to a complementary change in mixedness, the ``maximally coherent mixed states.'' The MCMS class consists of pseudo-pure states, which are an admixture of the maximally coherent pure state and an incoherent state. 
Incidentally, states of the form given by Eq.~(\ref{mxc}) have also been discussed as states of fixed purity that maximize 
the sum of quantum uncertainties \cite{Korzekwa2014}.

Similarly, one can derive a class of states with maximal mixedness for fixed coherence. Using an approach similar to Theorem \ref{t2}, one can show that the set of maximally mixed coherent states also satisfy the complementarity relation and thus lie on the parabola given by Eq. (\ref{conj}), and hence are of the same form as MCMS class.

\begin{figure}
\subfigure[~$d$ = 2]{
\includegraphics[height=42 mm, angle=-90]{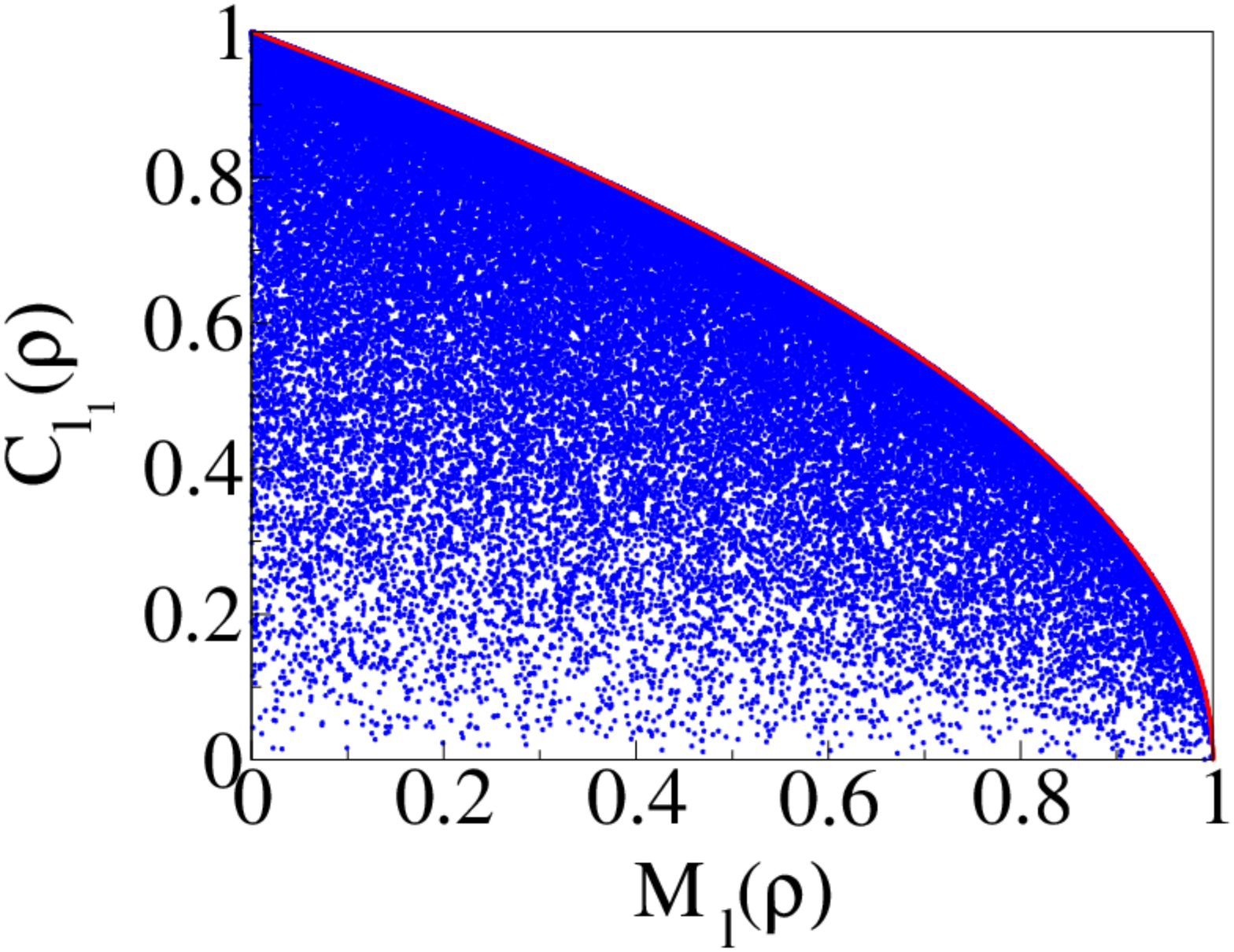}}
\subfigure[~$d$ = 3]{
\includegraphics[height=42 mm, angle=-90]{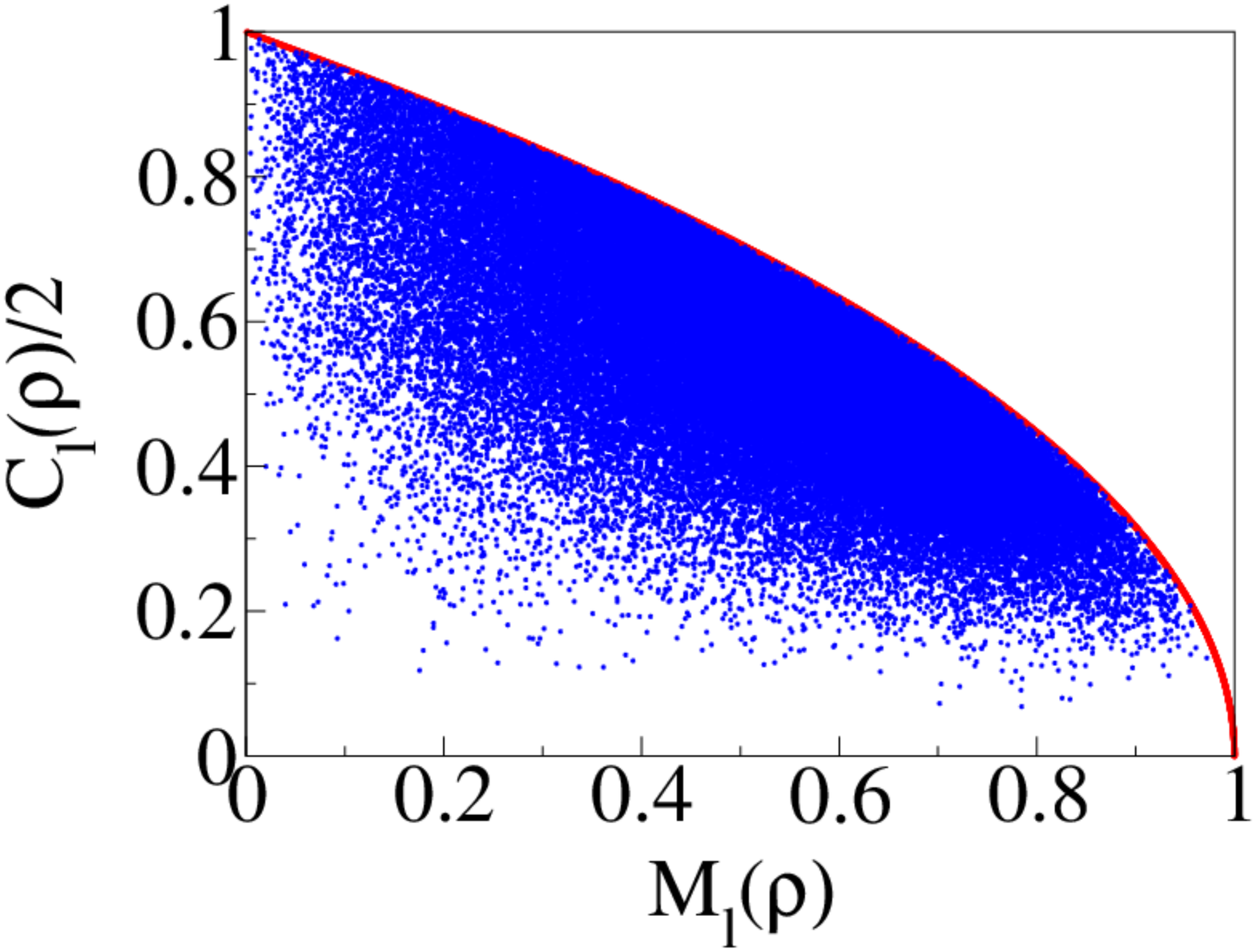}}\\
\subfigure[~$d$ = 4]{
\includegraphics[height=42 mm, angle=-90]{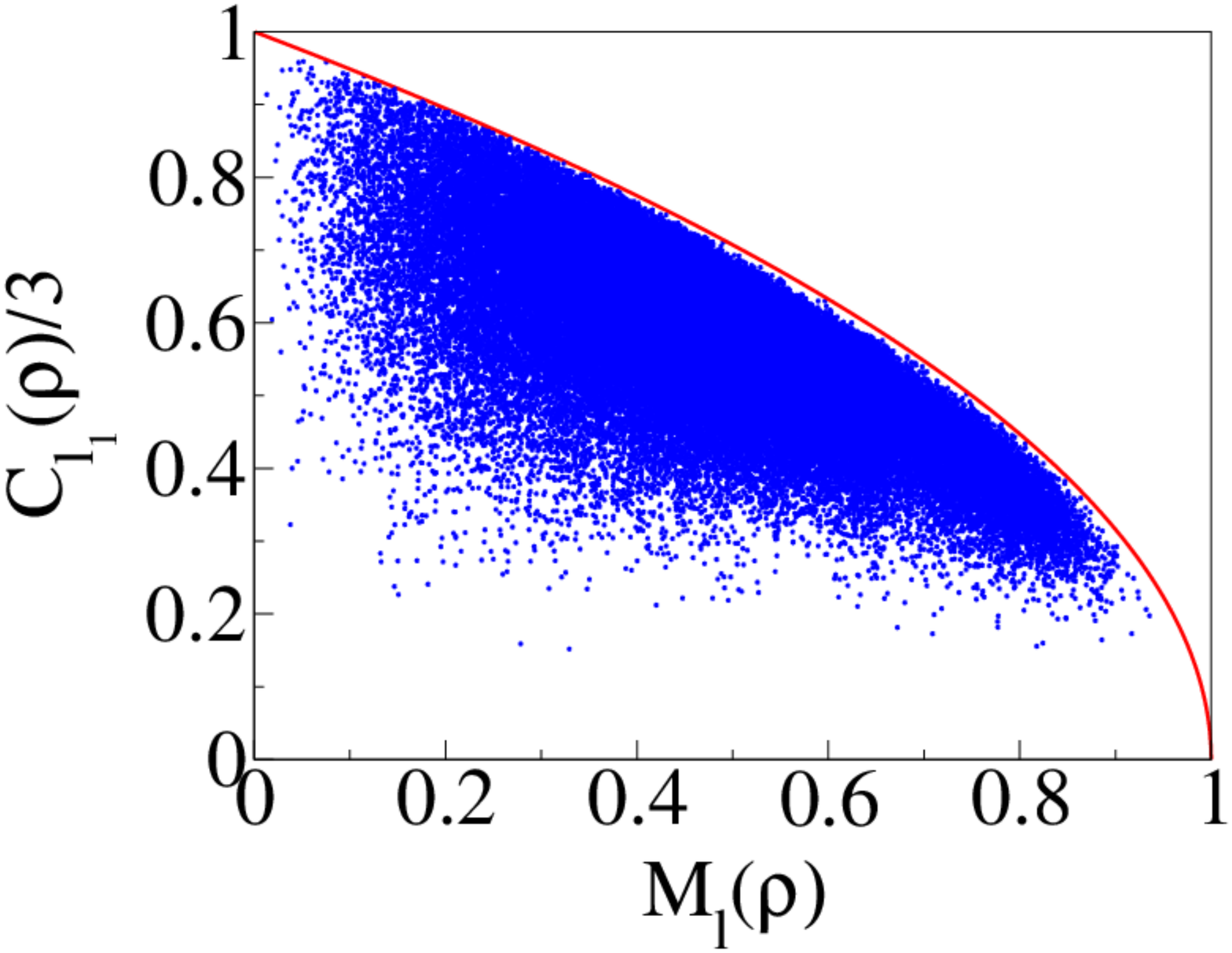}}
\subfigure[~$d$ = 5]{
\includegraphics[height=42 mm, angle=-90]{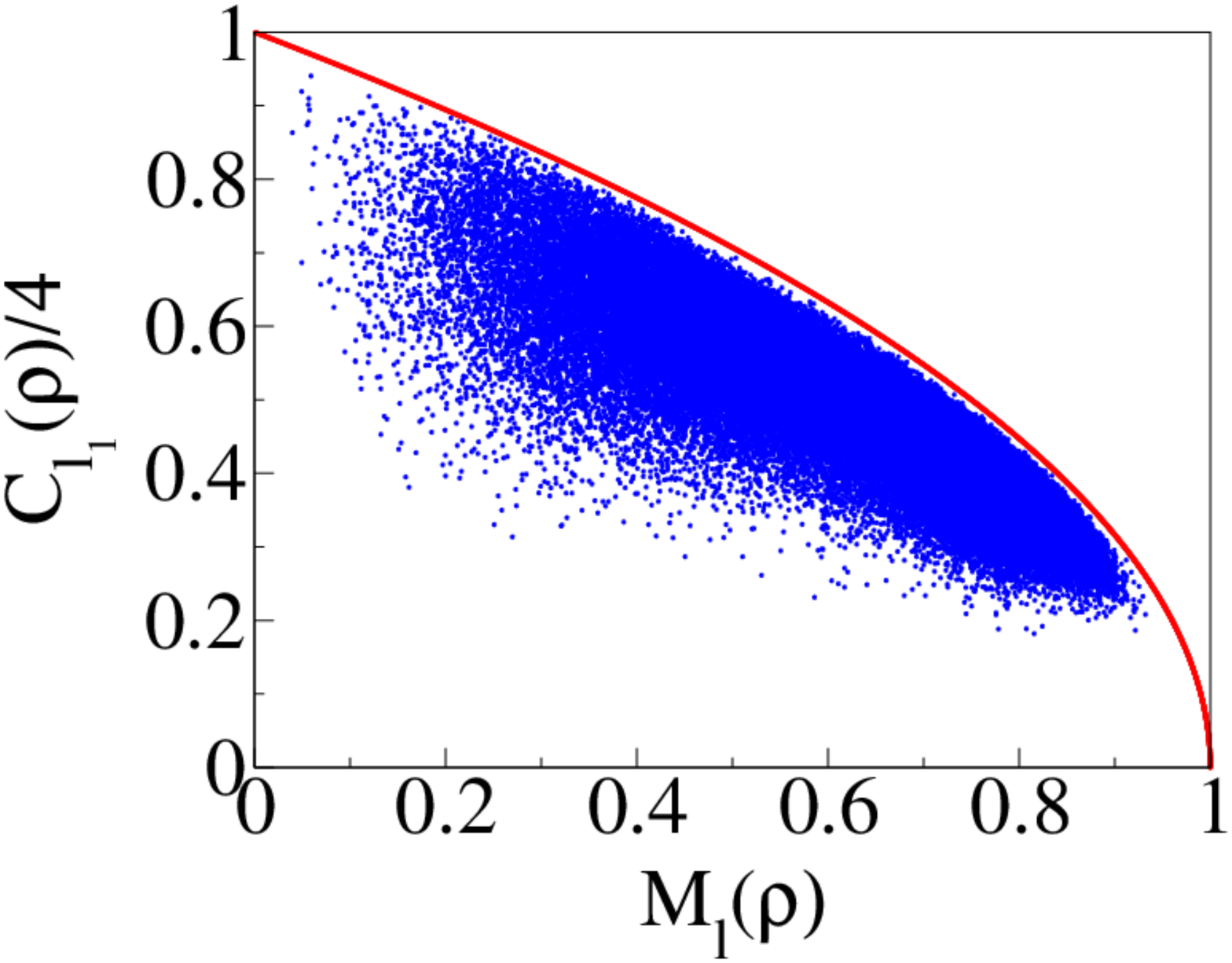}}
\caption{(Color online) Plot showing the trade-off between the (scaled) coherence, $C_{l_1}(\rho)/(d-1)$, and mixedness $M_l(\rho)$ as obtained from Eq. (\ref{Gconj}). The redline represents the extremal parabola in Eq. (\ref{conj}), which corresponds to the MCMS class that satisfies a complementarity relation between coherence and mixedness. The figure plots the (scaled) coherence, along the $Y$ axis, and mixedness, along the $X$ axis, for $1\times 10^5$ randomly generated states in $d$ = 2, 3, 4, and 5 dimensions, using a specific \emph{Mathematica} package \cite{Miszczak10}. 
}
\label{fig}
\end{figure}

Interestingly, we note that the form of MCMS remains the same if we employ a different set of measures 
for characterizing coherence and mixedness. For example, let us consider, the relative entropy of coherence $C_r(\rho)$ and von Neumann entropy $S(\rho)$ as our respective measures of coherence and mixedness. 
It can be shown, using the formalism employed in Theorems \ref{t1} and \ref{t2}, that the trade-off relation, $C_r(\rho) + S(\rho) \leq 1$, and the subsequent form of MCMS remains the same.
Similarly, if one considers geometric coherence and geometric mixedness for qubit systems, as the measures of coherence and mixedness, one can obtain an identical trade-off relation between the two quantities. To elaborate, the analytical form of geometric coherence for any arbitrary qubit state [Eq.(\ref{state})] is given by \cite{Alex15}, 
\begin{align}
\label{gcoh}
 C_{g}(\rho) = \frac{1}{2}\left[1-\sqrt{1-4|c|^2}\right],
\end{align}
where $c$ is the offdiagonal element of the qubit density matrix $\rho$ in the computational basis. Further, for an arbitrary qubit state, the geometric mixedness is given by
\begin{align}
\label{gmix}
 M_g(\rho) = \frac{1}{2}\left\{1+\sqrt{4\left[a(1-a)-|c|^2\right]}\right\}.
\end{align}
From Eqs. (\ref{gcoh}) and (\ref{gmix}), we have
\begin{align}
 C_{g}(\rho) + M_g(\rho) &= 1 + \frac{1}{2}\left\{\sqrt{4\left[a(1-a)-|c|^2\right]}- \sqrt{1-4|c|^2}\right\}\nonumber\\
 &\leq 1,
\end{align}
where in the last line we have used the fact that $4a(1-a)\leq 1$. Hence, we observe that the trade-off relation is the same in Theorem \ref{t1}.
For arbitrary qubit systems, the form of MCMS, given in Eq.~(\ref{mxc}), remains the same for the geometric coherence and geometric mixedness considered as the measures of coherence and mixedness, respectively, and the complementarity relation, $C_{g}(\rho) + M_g(\rho) = 1$, is satisfied. Hence there is a strong sense of universality about the form of MCMS, within the framework of the considered theory of coherence, in contrast to the measure dependent class of maximally entangled mixed states derived in the context of entanglement theory \cite{Ishizaka00, Munro01, Verstraete01, Verstraete03, Kwiat04}.
However, the form of MCMS for geometric coherence in general qudit systems needs to be further investigated. We note that the question of universality of the class MCMS for all equivalent sets of measures for coherence and mixedness, in any dimension, is still open.

\section{Transformations within classes of state}
\label{trans}

The trade-off between coherence and mixedness, as established in Theorem \ref{t1} along with the complementarity relation given by Eq. (\ref{conj}) for the MCMS class lead to the question of convertibility within the classes of fixed mixedness or coherence. In other words, given a class of states with fixed mixedness, what are the transformations that allow one to vary the coherence, while keeping the mixedness invariant, or vice versa? The importance of transformation and interconversion between classes of states 
lies in the predominant role it plays in resource theories \cite{Marvian2013, Marvian2014, FBrandao15} and its central status in the formulation of the second law(s) of thermodynamics in quantum regime \cite{Horodecki2013, Fernando2013, Varun14, Aberg14, Oppenheim14, Skrzypczyk2014, Brandao2015}.
In this section, we investigate the set of operations that allow for such transformations for qubit states. Here, we exclusively consider the $l_1$ norm of coherence and normalized linear entropy as the measures of coherence and mixedness, respectively.

\subsection{States with fixed coherence}

For a fixed value of coherence, say $\alpha$, in a fixed reference basis, say the computational basis, the states with varying mixedness, up to incoherent unitaries, are given by
$
 \rho(a) = \left(\begin{array}{cc}
 a & \alpha\\
 \alpha & 1-a
\end{array} \right).
$
Now, let us consider two states, $\rho(a_1)$ and $\rho(a_2)$, that have the same coherence but different mixedness. For the conditions,
$(1-a_1) \geq a_2 \geq a_1$ or $(1-a_1) \leq a_2 \leq a_1$, the inequality,  
$a_1 (1-a_1) \leq a_2 (1-a_2)$ is satisfied. 
For this case, it is easy to see that $\rho(a_2)$ is majorized \cite{Hardy88, Nielsen99, Nielsen01, Marshall2011, Wilde13} by $\rho(a_1)$, i.e., $\rho(a_2)\prec \rho(a_1)$.
Therefore, using Uhlmann's theorem \cite{Nielsen99, Nielsen01, Marshall2011, Wilde13}, we can write
\begin{align}
 \rho(a_2) = \sum_{i}p_i U_i \rho(a_1)U_i^\dagger,
\end{align}
where $U_i$'s are unitaries and $p_i\geq 0$, $\sum_i p_i=1$. For the qubit case, to keep the coherence invariant, we only allow incoherent unitaries. In the following,
we shall see that the map, 
\begin{align}
 \Phi[\rho] = p\rho + (1-p) \sigma_x\rho \sigma_x,
 \label{map}
\end{align}
where $\sigma_x=\left(\begin{array}{cc} 0 & 1\\ 1 & 0\end{array} \right)$, is sufficient to convert the state from $\rho(a_1)$ to $\rho(a_2)$, keeping the coherence unchanged. 
Specifically, we can achieve $\rho(a_2)$ from $\rho(a_1)$ using Eq. (\ref{map}), by setting $p=(1-a_1-a_2)/(1-2a_1)$,
 which is a valid probability for the case we are considering.
%
Similarly, in the opposite case with the conditions $(1-a_2) \geq a_1 \geq a_2$ or $(1-a_2) \leq a_1 \leq a_2$, one can find a similar map, as in Eq. (\ref{map}), from $\rho(a_2)$ to $\rho(a_1)$.

Therefore, given two qubit density matrices $\rho$ and $\sigma$ with the same coherence, if $\rho\prec\sigma$ ($\sigma\prec\rho$), then there will always exist a
probability distribution and incoherent unitaries, leading to a transformation  $\sigma \rightarrow \rho$ ($\rho \rightarrow \sigma$).
An interesting observation of the above analysis arises from considering maps related to open quantum systems. For noisy operations, for example the maps in Eq. (\ref{map}), the transformation between states with the same coherence is reminiscent of the phenomenon of freezing of quantum coherence \cite{Adesso2015}.

\subsection{States with fixed mixedness}

In the same vein, we explore the transformations which convert one state to the other with the same mixedness, but a varying amount of coherence. The states of the form
\begin{align}
 \rho(a) =\left(\begin{array}{cc}
 a & \sqrt{\frac{4a(1-a)-M}{4}}\\
 \sqrt{\frac{4a(1-a)-M}{4}} & 1-a
\end{array} \right),
\end{align}
have the same mixedness $M$ but can have different coherences. Now, let us consider two
different states $\rho(a_1)$ and $\rho(a_2)$.
Since these states have the same mixedness, and hence the same eigenvalues,
they must be related to each other by a unitary similarity transformation.
This similarity transformation can be easily found, once we get the eigenvectors of both the states. Let 
$\rho(a_1)\ket{e^{(1)}_i} = \lambda_i\ket{e^{(1)}_i}$ and $\rho(a_2)\ket{e^{(2)}_i} = \lambda_i\ket{e^{(2)}_i}$ ($i=1,2$).
Now, the unitary similarity transformation $S$, such that $\rho(a_2) = S\rho(a_1) S^\dagger$, can be obtained from the 
definition $S\ket{e^{(1)}_i}=\ket{e^{(2)}_i}$. Thus, for two states of given fixed mixedness, one can always find a reversible similarity
transformation between them. For an example, consider two states,
\begin{align}
 \rho_1 = \left(\begin{array}{cc}
 0.3 & 0.4\\
 0.4 & 0.7
\end{array} \right);~~~~~\rho_2 = \left(\begin{array}{cc}
 0.9 & 0.2\\
 0.2 & 0.1
\end{array} \right),
\end{align}
of the mixedness $M=0.2$. The similarity transformation from $\rho_2$ to $\rho_1$, i.e., $\rho_2=S\rho_1S^\dagger$,
using eigenvectors of both the states, is given by
$
 S = \frac{1}{\sqrt{2}} \left(\begin{array}{cc}
 1 & 1\\
 -1 & 1
\end{array} \right),
$
which is a coherent unitary.
In general, the states with identical mixedness but with varying coherence are connected through coherent unitaries.

\section{Conclusion}
\label{dis}

In our work, we show that there exists an intrinsic trade-off between the resourcefulness and the degree of noise in an arbitrary $d$-dimensional quantum system, as quantified by its coherence and mixedness, respectively. The obtained results are important from the perspective of resource theories as it allows us to quantify the maximal amount of coherence that can be harnessed from quantum states with a predetermined value of mixedness. Thus, we are able to analytically derive a class of maximally coherent mixed states, up to incoherent unitaries, that satisfy a complementarity relation between coherence and mixedness, in any quantum system. 
Due to the experimental ease with which the measurement of purity is feasible \cite{Ekert2002}, our results can be utilized to experimentally determine the maximal $l_1$ norm of coherence for any general $d$-dimensional quantum state. For qubit systems, the above conclusions can also be extended to the relative entropy and geometric measures of coherence.
Importantly, the theoretical formulation and results provided in the paper, are valid within the framework of the resource theory of coherence, as defined in \cite{Baumgratz2014}, and cannot be mathematically extended directly to the quantification of coherence based on the theory of asymmetry \cite{Marvian14}. Developing a framework that can operationally connect the two resource-theoretical perspectives is an important direction for future research.

The results presented in the work provide interesting insights on other aspects of the theory of coherence. An immediate application of our results 
is in understanding the connection between the resource theories of coherence and entanglement. It was shown in a recent paper \cite{Alex15}, that the maximum amount of entanglement
that can be created between a system and an incoherent ancilla, via incoherent operations, is equal to the coherence present in the system. 
Using the formalism presented in \cite{Alex15} and the complementarity relations derived in our work, one can prove that 
the maximum entanglement that can be created between a quantum system and an incoherent ancilla, via incoherent operations,
is bounded from above by the mixedness present in the system. 
Another significant aspect of the results is to address the question of order and interconvertibility between classes of quantum states, which is the 
fundamental premise for developing quantum resource theory and thermodynamics. Our analysis shows that, for qubit systems with a fixed coherence, majorization provides a total order on the states based on their degree of mixedness, while for fixed mixedness, all the qubit states with varying degree of coherences
are interconvertible. As a future direction, it will be very interesting to investigate if there exists such a total order in $d$-dimensional states with fixed coherence based on their degree of mixedness. 
We note that the total order on the states is only possible for a specific class of states and provided one works within the framework of the resource theory of coherence considered in our study. It is known that total order between states of fixed coherence is not possible within the resource theory of asymmetry \cite{Rudolph114, Gour2008, Marvian2013}.

To summarize, the present work deals with an important aspect of quantum physics, in particular, it addresses the question of how much a resource can be extracted from any arbitrary quantum system subjected to decoherence. We prove that there is a theoretical limit on the amount of coherence that can be extracted from mixed quantum systems and also derive the class of states that are most resourceful under decoherence. The results presented in the work provide impetus and alternative directions to the study of important physical quantities in open quantum systems and the effect of noise on quantum resources.

\begin{acknowledgments}
US, MNB, and HSD acknowledge the research fellowship of Department of Atomic Energy, Government of India.
\end{acknowledgments}

 \bibliographystyle{apsrev4-1}
 \bibliography{co-mix-lit}

\end{document}